\documentclass[aps,pra,reprint,superscriptaddress,longbibliography,nofootinbib]{revtex4-2}

\usepackage{amsmath,amssymb,mathtools,amsthm,bm,graphicx,booktabs,array}
\usepackage[nopatch=footnote]{microtype}
\usepackage[T1]{fontenc}
\usepackage{lmodern}
\usepackage{braket}

\newcommand{\id}{\operatorname{id}}
\newcommand{\tp}{\mathsf T}
\newcommand{\Tr}{\operatorname{Tr}}
\newcommand{\rank}{\operatorname{rank}}

\newcommand{\spec}{\operatorname{spec}}

\newcommand{\Sym}{\operatorname{Sym}}

\newcommand{\diag}{\operatorname{diag}}

\newcommand{\desc}{\downarrow}
\newcommand{\cD}{\mathcal D}
\newcommand{\cE}{\mathcal E}
\newcommand{\cF}{\mathcal F}
\newcommand{\cH}{\mathcal H}
\newcommand{\cN}{\mathcal N}
\newcommand{\cS}{\mathcal S}

\newcommand{\cC}{\mathcal C}

\newcommand{\bbC}{\mathbb C}
\newcommand{\norm}[1]{\left\lVert #1\right\rVert}
\newcommand{\abs}[1]{\left\lvert #1\right\rvert}

\newtheorem{theorem}{Theorem}
\newtheorem{proposition}[theorem]{Proposition}

\newtheorem{corollary}[theorem]{Corollary}
\newtheorem{definition}[theorem]{Definition}

\graphicspath{{figures/}}

\begin{document}

\title{The inverse eigenvalue problem for Partial Transposes of Two-Qubit States
 and Sharp Entanglement Bounds from Their Low-order Moments}

\author{Ruoting Dou}
\affiliation{National Laboratory of Solid State Microstructures and School of Physics, Collaborative Innovation Center of Advanced Microstructures, Nanjing University, Nanjing 210093, China}

\author{Shengjun Wu}
\email{sjwu@nju.edu.cn}
\affiliation{National Laboratory of Solid State Microstructures and School of Physics, Collaborative Innovation Center of Advanced Microstructures, Nanjing University, Nanjing 210093, China}

\author{Zeng-Bing Chen}
\email{zbchen@nju.edu.cn}
\affiliation{National Laboratory of Solid State Microstructures and School of Physics, Collaborative Innovation Center of Advanced Microstructures, Nanjing University, Nanjing 210093, China}

\begin{abstract}
The negativity of a finite-dimensional bipartite state $\rho$ can be determined
without full state tomography by measuring sufficiently many partial-transpose
(PT) moments on multiple copies of $\rho$. When the available copies allow only low-order PT moments to be estimated reliably, how tightly do these moments bound the negativity?
For two-qubit states, we answer this question by solving the inverse
eigenvalue problem for their partial transposes: an ordered eigenvalue list $(x,y,z,-q)$, with $x\ge y\ge z\ge0$, $q>0$, and $x+y+z-q=1$, occurs as the spectrum of $\rho^{\Gamma_B}$ for some two-qubit state $\rho$ if and only if $q\le y$ and $qy\le xz$. The quantum steering ellipsoid of the $X$ state constructed to prove
sufficiency gives a geometric interpretation of the two equality cases. We use these inequalities to determine the exact range of negativity compatible with the two lowest nontrivial PT moments, $p_2=\Tr[(\rho^{\Gamma_B})^2]$ and $p_3=\Tr[(\rho^{\Gamma_B})^3]$. We also compare our sharp negativity bounds with estimates from several
existing methods and quantify how much the inverse eigenvalue constraints narrow the negativity interval as a function of $p_3$ at fixed $p_2$. Beyond negativity, we extend this approach to derive attainable bounds
for the maximal CHSH value
and directional coherent information at fixed $p_2$ and $p_3$. Our result establishes the strongest possible bounds on the negativity
of an unknown two-qubit state from just two experimentally accessible
PT moments, and shows the potential for obtaining similarly sharp bounds from
low-order PT moments by solving the corresponding inverse eigenvalue
problem for higher-dimensional bipartite states.
\end{abstract}

\maketitle

\section{Introduction}
\label{sec:intro}

Detecting and quantifying entanglement are central tasks in quantum information. Partial transposition provides one of the simplest spectral tests for bipartite entanglement. If a finite-dimensional bipartite state $\rho$ is separable, then $\rho^{\Gamma_B}\succeq0$; the converse holds for $2\otimes2$ and $2\otimes3$ systems \cite{Peres1996,Horodecki1996}. The negative eigenvalues of $\rho^{\Gamma_B}$ define the negativity
\begin{equation}
 \cN(\rho)=\frac{\norm{\rho^{\Gamma_B}}_1-1}{2},
 \label{eq:negativity-intro}
\end{equation}
which is a standard mixed-state entanglement monotone, as is the logarithmic negativity
\cite{VidalWerner2002,Plenio2005,HorodeckiRMP2009}.
The full spectrum of $\rho^{\Gamma_B}$ determines the partial-transpose (PT) moments
\begin{equation}
 p_k=\Tr[(\rho^{\Gamma_B})^k].
 \label{eq:pt-moments-intro}
\end{equation}
Conversely, sufficiently many PT moments determine the eigenvalues of
$\rho^{\Gamma_B}$ and hence its negativity~\cite{Carteret2005}.
For two qubits, $p_2,p_3,p_4$ suffice because
$p_1=\Tr(\rho^{\Gamma_B})=1$~\cite{Bartkiewicz2015}.
These moments can be estimated without full state tomography using
multicopy circuits, local randomized measurements, or sequential qubit
reuse~\cite{Carteret2005,Elben2020,Neven2021,HuangEtAl2026}.
However, in the practical experiment system, the number of available copies is often limited, thus only low-order PT moments may be estimated reliably. It is then a central question that how much negativity do limit moments information guarantee, and what is the largest negativity compatible with them?

For two qubits with known $p_2$ and $p_3$, prior work addresses
several closely related questions. Yu, Imai, and G\"uhne characterized the moment pairs $(p_2,p_3)$
compatible with a separable state, thereby identifying exactly when
the minimum two-qubit negativity is zero~\cite{YuImaiGuhne2021}.
The region of $(p_2,p_3)$ for two-qubit states has been studied
in Refs.~\cite{Zhang2022,Zhang2023}, and PT moments have also been
used to study entanglement phase diagrams~\cite{Carrasco2024}.
Carteret derived a lower bound on the negativity from $p_2$ and
$p_3$~\cite{Carteret2017}. Lower bounds on the logarithmic negativity
have also been obtained from even PT moments~\cite{TarabungaHaug2026,
MillerEisert2026,ChakrabartyMajumdar2026}.
Despite this substantial progress, none of these studies determines
the full range of negativity for two-qubit states at every admissible
pair $(p_2,p_3)$.

To obtain sharp bounds on the negativity at fixed $p_2$ and $p_3$, it
is essential to determine which eigenvalue lists consistent with these
moments can occur as $\operatorname{spec}(\rho^{\Gamma_B})$ for some two-qubit
state $\rho$. That is the so-called \emph{inverse eigenvalue problem}: given a trace-one eigenvalue list $\boldsymbol{\lambda}$, does there exist a density operator $\rho_{AB}$ such that
\begin{equation}
\boldsymbol{\lambda}
=
\operatorname{spec}^{\downarrow}\!\left(\rho_{AB}^{\Gamma_B}\right)?
\end{equation} 
Related inverse problems include the inverse eigenvalue problem for block-positive operators \cite{JohnstonPatterson2018} and absolute-PPT theory \cite{Hildebrand2007}. We solve the inverse problem completely for two qubits. Every nonnegative trace-one spectrum is realized as $\operatorname{spec}(\rho^{\Gamma_B})$ by a diagonal PPT state. Since the partial transpose of a two-qubit state has at most one negative eigenvalue \cite{Sanpera1998,Rana2013}, any ordered candidate eigenvalue list with a negative entry can be written as 
\begin{equation}
 \boldsymbol{\lambda}=(x,y,z,-q),\qquad
 x\ge y\ge z\ge0,\quad q>0,
 \label{eq:NPT-candidate-intro}
\end{equation}
with $x+y+z-q=1$. Our main result is
\begin{equation}
\boldsymbol{\lambda}\ \text{is realizable}\quad\Longleftrightarrow\quad
q\le y,\qquad qy\le xz.
\label{eq:main-iff-intro}
\end{equation}
Necessity follows from the local symmetric--antisymmetric decomposition
on two copies of $\rho^{\Gamma_B}$. An explicit $X$ state proves
sufficiency. For comparison, block positivity requires $q\le y$ and
the weaker product inequality $q^2\le xz$~\cite[Theorem~3]{JohnstonPatterson2018}.
The quantum steering ellipsoid of the constructed $X$ state gives a
geometric interpretation of the two equality cases~\cite{Jevtic2014QSE,Milne2014ExtremalQSE}.
Its center coordinate $c=(y-q)/(1-z)$ vanishes precisely when $q=y$,
and its volume ratio $V/V_{\max}(c)=qy/(xz)$ reaches unity precisely
when $qy=xz$, as shown in Figure~\ref{fig:spectral-QSE-main}.

At fixed $p_2$ and $p_3$, the sharp negativity bounds follow by
minimizing and maximizing $q$ over the ordered lists with these
moments that satisfy Eq.~\eqref{eq:main-iff-intro}. Zero is included
whenever the moments admit a PPT state.
Figure~\ref{fig:moment-range-main} compares the resulting intervals
with estimates from several existing methods as $p_3$ varies at fixed
$p_2$. Block positivity gives the exact upper bound. The stronger
condition $qy\le xz$ raises the lower bound in an explicitly
characterized region of the $(p_2,p_3)$ plane. We also provide a practical example of how to apply the sharp negativity
bounds. Consider a single excitation shared between
two qubits and subject to loss. At a survival probability of $8/9$,
the moments $p_2$ and $p_3$ constrain the negativity of any compatible
two-qubit state to an interval $56.8\%$ narrower than the one obtained
from block positivity alone (Table~\ref{tab:loss}). Beyond negativity, we apply the same approach to compute the attainable bounds on the maximal CHSH value and the directional coherent information
at fixed $p_2$ and $p_3$. The two-qubit result shows the potential
for obtaining similarly sharp bounds from low-order PT moments by
solving the corresponding inverse eigenvalue problem in higher dimensions.

\section{The inverse eigenvalue problem and the two-copy structure}
\label{sec:two-copy}

\subsection{Inverse eigenvalue formulation}

Let $A\simeq\mathbb C^m$ and $B\simeq\mathbb C^n$, with $\cD(A\otimes B)$ denoting the density operators on $A\otimes B$. We write $\Gamma_B=\id_A\otimes\tp_B$ for partial transposition on $B$ in a fixed local basis. The eigenvalues of $\rho^{\Gamma_B}$ are independent of this basis.
Since $\rho^{\Gamma_A}=(\rho^{\Gamma_B})^{\tp}$, transposing either
subsystem gives the same eigenvalues. We omit the subsystem label
when it is unambiguous.

\begin{definition}[Partial-transpose inverse-eigenvalue set]
For ordered eigenvalue lists, define
\begin{equation}
 \cS_{\rm PT}^{m,n}
 :=
 \left\{
 \spec^{\desc}(\rho^{\Gamma_B}):
 \rho\in\cD(\mathbb C^m\otimes\mathbb C^n)
 \right\}.
 \label{eq:SPT-def-main}
\end{equation}
The partial-transpose inverse eigenvalue problem is to characterize $\cS_{\rm PT}^{m,n}$.
\end{definition}

A useful comparison is provided by block-positive operators. Let $\cS_{\rm BP,1}^{m,n}$ denote the ordered spectra of trace-one Hermitian operators $W$ satisfying
\begin{equation}
 \bra{a,b}W\ket{a,b}\ge0
 \label{eq:block-positive-def-main}
\end{equation}
for every product vector $\ket{a,b}$. Every partial transpose of a density operator is block positive, since
\begin{equation}
 \bra{a,b}\rho^{\Gamma_B}\ket{a,b}
 =
 \bra{a,b^*}\rho\ket{a,b^*}
 \ge0.
 \label{eq:block-positivity-main}
\end{equation}
Hence
\begin{equation}
 \cS_{\rm PT}^{m,n}\subseteq\cS_{\rm BP,1}^{m,n}.
 \label{eq:SPT-in-BP-main}
\end{equation}
The inclusion need not be an equality: block positivity of $W$ does not require $W$ itself to arise as the partial transpose of a density operator.

The nonnegative part of $\cS_{\rm PT}^{m,n}$ follows immediately. For any ordered nonnegative eigenvalue list
\begin{equation}
 \lambda_1\ge\cdots\ge\lambda_{mn}\ge0,
 \qquad
 \sum_{i=1}^{mn}\lambda_i=1,
 \label{eq:PPT-simplex-main}
\end{equation}
the product-basis diagonal state
\begin{equation}
 \rho=\diag(\lambda_1,\ldots,\lambda_{mn})
 \label{eq:PPT-realization-main}
\end{equation}
satisfies $\rho^{\Gamma_B}=\rho$. The inverse problem is therefore nontrivial only for candidate eigenvalue lists with a negative entry.

For two qubits, an ordered candidate eigenvalue list with one negative entry can be written as
\begin{equation}
 \boldsymbol{\lambda}=(x,y,z,-q),
 \qquad
 x\ge y\ge z\ge0,\quad q>0,
 \label{eq:twoqubit-spectrum-main}
\end{equation}
with
\begin{equation}
 x+y+z-q=1.
 \label{eq:twoqubit-trace-main}
\end{equation}
The corresponding inverse eigenvalue theorem for block-positive operators requires only
\begin{equation}
 \boldsymbol{\lambda}\in\cS_{\rm BP,1}^{2,2}
 \quad\Longleftrightarrow\quad
 q\le y,\qquad q^2\le xz
 \label{eq:BP-characterization-main}
\end{equation}
\cite[Theorem~3]{JohnstonPatterson2018}.

\subsection{The two-copy antisymmetric subspace}

For a bipartite density operator $\rho\in\cD(A\otimes B)$, write $\tau=\rho^{\Gamma_B}$ and $\cH=A\otimes B$. To derive constraints on the spectrum of $\tau$ from positivity of $\rho$, consider $\tau\otimes\tau$ and decompose its globally antisymmetric sector according to the exchange symmetries of the two local subsystems. Use the canonical regrouping
\begin{equation}
 (A_1\otimes B_1)\otimes(A_2\otimes B_2)
 \simeq
 (A_1\otimes A_2)\otimes(B_1\otimes B_2).
 \label{eq:regroup-main}
\end{equation}
Let $F_A$ and $F_B$ exchange the two copies of $A$ and $B$, respectively, and define the corresponding local parity projectors
\begin{equation}
 P_\pm^A=\frac{I\pm F_A}{2},
 \qquad
 P_\pm^B=\frac{I\pm F_B}{2}.
 \label{eq:local-parity-main}
\end{equation}
The swap of the two complete bipartite copies is $F_{\cH}=F_AF_B$, so
\begin{equation}
 P_-^{\cH}=\frac{I-F_AF_B}{2}
 \label{eq:global-antisym-main}
\end{equation}
projects onto the globally antisymmetric subspace $\wedge^2\cH$.

Global antisymmetry requires opposite local exchange parities: one local pair is antisymmetric while the other is symmetric. Introduce
\begin{align}
 Q_E&=P_-^A\otimes P_+^B,
 &
 \cE&=\wedge^2A\otimes\Sym^2B,
 \label{eq:QE-main}\\
 Q_F&=P_+^A\otimes P_-^B,
 &
 \cF&=\Sym^2A\otimes\wedge^2B.
 \label{eq:QF-main}
\end{align}
The projectors are orthogonal and satisfy $Q_E+Q_F=P_-^{\cH}$, which gives
\begin{equation}
 \wedge^2(A\otimes B)=\cE\oplus\cF.
 \label{eq:skew-Cauchy-main}
\end{equation}
This is the degree-two skew Cauchy decomposition \cite{FultonHarris1991}. On $\wedge^2\cH$ define
\begin{equation}
 J=Q_E-Q_F.
 \label{eq:J-main}
\end{equation}
The operator $J$ is Hermitian and unitary, acting as $+1$ on $\cE$ and $-1$ on $\cF$.

Because $\tau\otimes\tau$ commutes with the complete-copy swap $F_{\cH}$, it preserves $\wedge^2\cH$. Denote its restriction to this subspace by
\begin{equation}
 K
 :=
 P_-^{\cH}(\tau\otimes\tau)P_-^{\cH}
 \big|_{\wedge^2\cH}
 =
 \wedge^2\tau .
 \label{eq:K-main}
\end{equation}
If $\tau\ket{v_i}=\lambda_i\ket{v_i}$, then the antisymmetric vector
\begin{equation}
 \ket{v_i}\wedge\ket{v_j}
 :=
 \frac{
 \ket{v_i}\ket{v_j}-\ket{v_j}\ket{v_i}
 }{\sqrt2}
\end{equation}
satisfies
\begin{equation}
 K(\ket{v_i}\wedge\ket{v_j})
 =
 \lambda_i\lambda_j
 (\ket{v_i}\wedge\ket{v_j}),
 \qquad i<j.
 \label{eq:wedge-spectrum-main}
\end{equation}
Thus
\begin{equation}
 \spec(K)=\{\lambda_i\lambda_j:i<j\}.
 \label{eq:K-spectrum-main}
\end{equation}

To connect these products with positivity of the original state, define
\begin{equation}
 S_E=Q_E\rho^{\otimes2}Q_E,
 \qquad
 S_F=Q_F\rho^{\otimes2}Q_F.
 \label{eq:SE-SF-main}
\end{equation}
Both operators are positive semidefinite. Here and below, partial transposition on the regrouped two-copy space acts on the full factor $B_1\otimes B_2$.

\begin{proposition}[Two-copy parity-block identity]
\label{prop:parity-block-main}
For every bipartite density operator,
\begin{equation}
 K+JKJ
 =
 2\left(
 S_E^{\Gamma_{B_1B_2}}
 \oplus
 S_F^{\Gamma_{B_1B_2}}
 \right).
 \label{eq:block-identity-main}
\end{equation}
\end{proposition}

\begin{proof}
Relative to the decomposition $\wedge^2\cH=\cE\oplus\cF$, conjugation by $J$ leaves the diagonal blocks of $K$ unchanged and reverses the sign of the off-diagonal blocks. Hence
\begin{equation}
 K+JKJ
 =
 2(Q_EKQ_E+Q_FKQ_F).
 \label{eq:J-pinching-main}
\end{equation}
Moreover,
\begin{equation}
 \tau\otimes\tau
 =
 (\rho^{\otimes2})^{\Gamma_{B_1B_2}}.
\end{equation}
The projectors $P_\pm^B$ are real and symmetric and are therefore invariant under transposition. Compression by $Q_E$ or $Q_F$ consequently commutes with the partial transpose, giving
\begin{equation}
 Q_EKQ_E=S_E^{\Gamma_{B_1B_2}},
 \qquad
 Q_FKQ_F=S_F^{\Gamma_{B_1B_2}}.
 \label{eq:compression-PT-main}
\end{equation}
Substitution into Eq.~\eqref{eq:J-pinching-main} proves Eq.~\eqref{eq:block-identity-main}.
\end{proof}

Proposition~\ref{prop:parity-block-main} holds in arbitrary local dimensions. The simplification specific to two qubits follows from the one-dimensionality of the local antisymmetric subspaces.

\begin{theorem}[Two-copy positivity and eigenvalue pairing]
\label{thm:two-copy-pairing-main}
Let $\rho$ be a two-qubit density operator and $K=\wedge^2(\rho^{\Gamma_B})$. Then
\begin{equation}
 K+JKJ\succeq0.
 \label{eq:KJK-positive-main}
\end{equation}
If
\begin{equation}
 \kappa_1\ge\kappa_2\ge\cdots\ge\kappa_6
\end{equation}
are the eigenvalues of $K$, then
\begin{equation}
 \kappa_i+\kappa_{7-i}\ge0,
 \qquad i=1,2,3.
 \label{eq:kappa-pairing-main}
\end{equation}
\end{theorem}

\begin{proof}
For $A\simeq B\simeq\mathbb C^2$,
\begin{equation}
 \dim\wedge^2A=\dim\wedge^2B=1.
\end{equation}
Let $\ket{a_-}$ and $\ket{b_-}$ be normalized generators of these two antisymmetric spaces. Since $S_E,S_F\succeq0$, there exist $X_B,Y_A\succeq0$ such that
\begin{equation}
\begin{aligned}
 S_E &= \ket{a_-}\!\bra{a_-}\otimes X_B,\\
 S_F &= Y_A\otimes\ket{b_-}\!\bra{b_-}.
\end{aligned}
\label{eq:factor-SE-SF-main}
\end{equation}
They remain positive semidefinite after partial transposition, and Proposition~\ref{prop:parity-block-main} therefore gives Eq.~\eqref{eq:KJK-positive-main}.

The latter inequality is equivalent to $K\succeq-JKJ$. Since $J$ is unitary, $-JKJ$ has ordered eigenvalues $-\kappa_6,\ldots,-\kappa_1$; Loewner monotonicity then gives $\kappa_i\ge-\kappa_{7-i}$, i.e., Eq.~\eqref{eq:kappa-pairing-main}.
\end{proof}

\section{Two-qubit inverse eigenvalue theorem and steering geometry}
\label{sec:two-qubit-qse}

The two-qubit inverse problem is now settled by the
following complete characterization.
\begin{theorem}[Complete two-qubit inverse eigenvalue theorem]
\label{thm:inverse-main}
Let $\boldsymbol{\lambda}=(\lambda_1,\lambda_2,\lambda_3,\lambda_4)$ satisfy
\[
\lambda_1\ge\lambda_2\ge\lambda_3\ge\lambda_4,\qquad
\sum_{i=1}^4\lambda_i=1.
\]
There exists a two-qubit density operator $\rho$ such that
\begin{equation}
 \spec^{\desc}(\rho^{\Gamma_B})
 =\boldsymbol{\lambda}
 \label{eq:inverse-theorem-spectrum-main}
\end{equation}
if and only if either $\lambda_4\ge0$, or
\begin{equation}
 \boldsymbol{\lambda}=(x,y,z,-q),
 \qquad
 x\ge y\ge z\ge0,\quad q>0,
 \label{eq:ordered-NPT-main}
\end{equation}
with
\begin{equation}
 q\le y,\qquad qy\le xz.
 \label{eq:inverse-conditions-main}
\end{equation}
\end{theorem}

\begin{proof} 
A nonnegative trace-one spectrum is realized by a product-basis diagonal state. Suppose now that the candidate eigenvalue list in Eq.~\eqref{eq:ordered-NPT-main} is realized as $\operatorname{spec}(\rho^{\Gamma_B})$ for an NPT state $\rho$. The ordered eigenvalues of $K=\wedge^2(\rho^{\Gamma_B})$ are

\begin{equation}
 xy,\quad xz,\quad yz,\quad -qz,\quad -qy,\quad -qx.
 \label{eq:wedge-eigenvalues-main}
\end{equation}
The pairing inequalities in Theorem~\ref{thm:two-copy-pairing-main} therefore give
\begin{equation}
 x(y-q)\ge0,\qquad
 xz-qy\ge0,\qquad
 z(y-q)\ge0.
 \label{eq:necessity-three-main}
\end{equation}
Since $x>0$, the first two yield $q\le y$ and $qy\le xz$; the third is then redundant.

Conversely, assume Eq.~\eqref{eq:inverse-conditions-main} and define
\begin{equation}
 \rho_X=
 \begin{pmatrix}
 0&0&0&0\\
 0&x&\sqrt{qy}&0\\
 0&\sqrt{qy}&z&0\\
 0&0&0&y-q
 \end{pmatrix}
 \label{eq:X-state-main}
\end{equation}
in the product basis $\{\ket{00},\ket{01},\ket{10},\ket{11}\}$. The inequalities $q\le y$ and $qy\le xz$ imply $\rho_X\succeq0$, and the trace condition gives $\Tr\rho_X=1$. Its partial transpose is
\begin{equation}
 \rho_X^{\Gamma_B}=
 \begin{pmatrix}
 0&0&0&\sqrt{qy}\\
 0&x&0&0\\
 0&0&z&0\\
 \sqrt{qy}&0&0&y-q
 \end{pmatrix}.
 \label{eq:X-state-PT-main}
\end{equation}
The corner block has characteristic polynomial
\[
 \mu^2-(y-q)\mu-qy=(\mu-y)(\mu+q),
\]
so $\spec^{\desc}(\rho_X^{\Gamma_B}) =(x,y,z,-q)=\boldsymbol{\lambda}$.
\end{proof}

Figure~\ref{fig:spectral-QSE-main} summarizes the spectral and geometric boundaries. Combining Theorem~\ref{thm:inverse-main} with the block-positive characterization in Eq.~\eqref{eq:BP-characterization-main} gives the strict inclusion $\cS_{\rm PT}^{2,2}\subsetneq\cS_{\rm BP,1}^{2,2}$. For example, $(x,y,z,q) = \left( \frac{21}{25}, \frac15, \frac1{100}, \frac1{20} \right)$ satisfies $q^2=\frac1{400} \le \frac{21}{2500}=xz < \frac1{100}=qy$, as well as $q\le y$. It therefore satisfies the complete block-positive conditions but violates the second inverse-eigenvalue inequality.

\begin{figure*}[t]
    \centering
    \includegraphics[width=\textwidth]{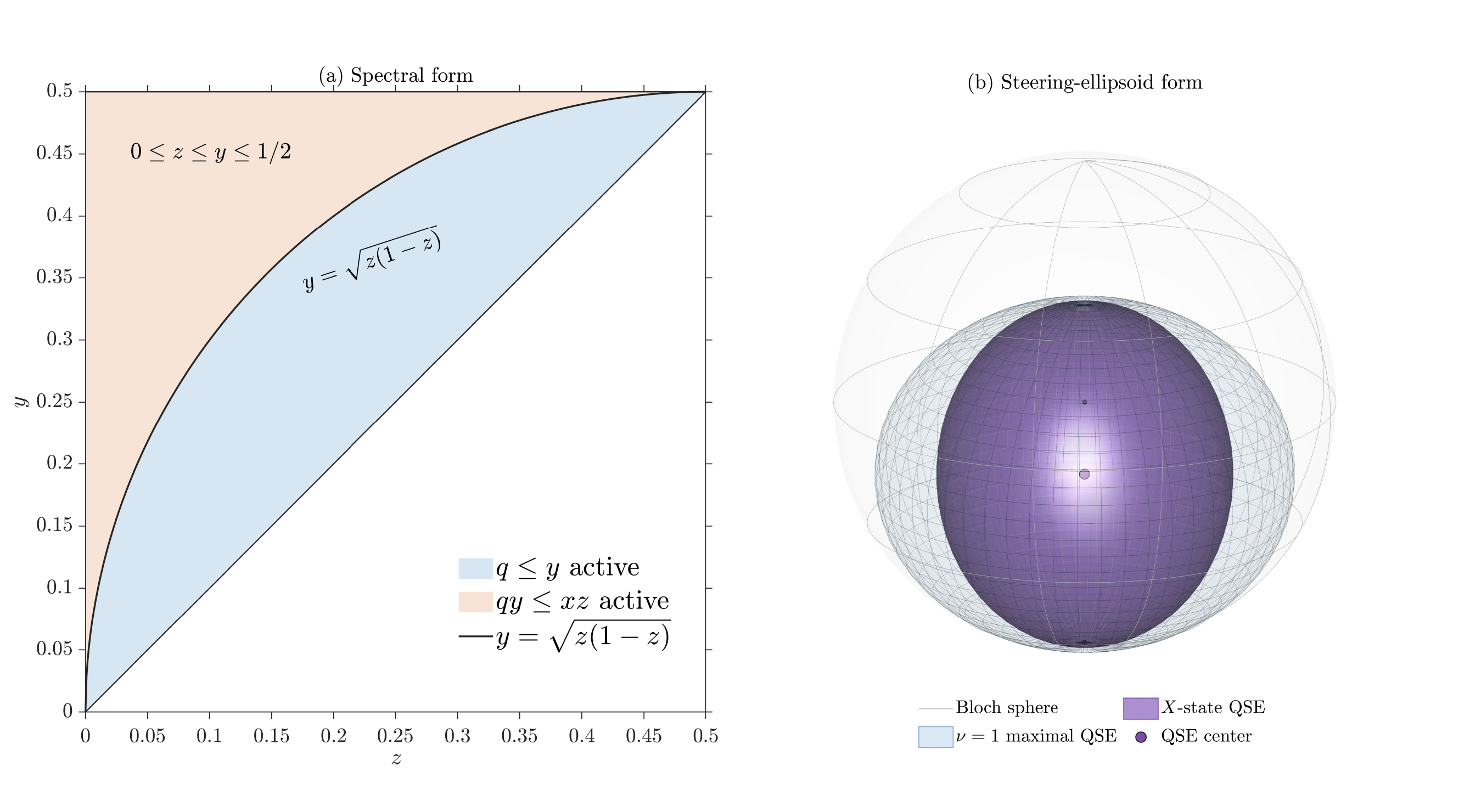}
    \caption{
Spectral and geometric forms of the two-qubit inverse eigenvalue theorem.
(a)~For an ordered candidate eigenvalue list $\boldsymbol{\lambda}=(x,y,z,-q)$ with one negative entry, the projected domain $0\le z\le y\le1/2$ is divided according to whether $q\le y$ or $qy\le xz$ sets the upper bound on $q$; the two bounds meet at $y=\sqrt{z(1-z)}$.
(b)~The $X$ state used to establish sufficiency has an axisymmetric steering ellipsoid tangent to the south pole, centered at $c=(y-q)/(1-z)$, with $s_{\parallel}=1-c$ and $s_{\perp}^{2}=\nu(1-c)$. Here $\nu=qy/(xz)=V_X/V_{\max}(c)$, so $\nu=1$ gives the maximal-volume steering ellipsoid allowed at fixed center and corresponds to the boundary $qy=xz$.
}
    \label{fig:spectral-QSE-main}
\end{figure*}

\subsection{Steering geometry of the $X$-state representative}
\label{subsec:QSE-main}

Equation~\eqref{eq:inverse-conditions-main} has the following Bloch-geometric expression. Writing subsystems $A$ and $B$ as Alice and Bob, respectively, use the Fano representation
\begin{equation}
 \begin{aligned}
 \rho=\frac14\biggl[&I\otimes I
 +\bm r\cdot\bm\sigma\otimes I
 +I\otimes\bm s\cdot\bm\sigma\\
 &+\sum_{i,j=1}^{3}M_{ij}\sigma_i\otimes\sigma_j
 \biggr].
 \end{aligned}
 \label{eq:Fano-main}
\end{equation}
For a measurement effect on Bob with Bloch vector $\bm n$, Alice's conditional Bloch vector is
\begin{equation}
 \bm a(\bm n)
 =
 \frac{\bm r+M\bm n}{1+\bm s\cdot\bm n},
 \label{eq:conditional-Bloch-main}
\end{equation}
where $|\bm n|\le1$ parametrizes Bob's positive measurement effects up to scale. These conditional vectors fill the quantum steering ellipsoid~\cite{Jevtic2014QSE,Zhang2019ExperimentalQSE,Xu2024EllipsoidZoo}. Defining
\begin{equation}
 G=M-\bm r\bm s^{\mathsf T},
 \label{eq:connected-tensor-main}
\end{equation}
its center and shape matrix are
\begin{align}
 \bm c_A&=\frac{\bm r-M\bm s}{1-\abs{\bm s}^2},
 \label{eq:QSE-center-main}\\
 Q_A&=\frac{1}{1-\abs{\bm s}^2}
 G\left(I+\frac{\bm s\bm s^{\mathsf T}}{1-\abs{\bm s}^2}\right)G^{\mathsf T}.
 \label{eq:QSE-shape-main}
\end{align}
The semiaxes are the square roots of the eigenvalues of $Q_A$.

For the state in Eq.~\eqref{eq:X-state-main}, direct evaluation gives
\begin{align}
 \bm r_X&=(0,0,2x-1)^{\mathsf T},\nonumber\\
 \bm s_X&=(0,0,2z-1)^{\mathsf T},
 \label{eq:X-local-Bloch-main}\\
 M_X&=\diag\!\bigl(2\sqrt{qy},2\sqrt{qy},
 2(y-q)-1\bigr).
 \label{eq:X-correlation-main}
\end{align}
The ellipsoid is centered at
\begin{equation}
 \bm c_A=(0,0,-c)^{\mathsf T},
 \qquad
 c=\frac{y-q}{1-z},
 \label{eq:c-definition-main}
\end{equation}
and has semiaxes
\begin{equation}
 s_1=s_2=s_\perp,
 \qquad
 s_\perp^2=\frac{qy}{z(1-z)},
 \label{eq:transverse-axis-main}
\end{equation}
\begin{equation}
 s_3=s_\parallel=\frac{x}{1-z}=1-c.
 \label{eq:longitudinal-axis-main}
\end{equation}
Using $x+y+z-q=1$, we find $s_{\parallel}=1-c$, so the lower end of the ellipsoid is $-c-s_{\parallel}=-1$. Each such $X$ state therefore defines an axisymmetric steering spheroid tangent to the south pole of the Bloch sphere.

For a fixed center $c$, Milne et al. \cite{Milne2014ExtremalQSE,Milne2015Corrigendum} showed that the maximal-volume steering ellipsoid allowed at that center has semiaxes
\begin{equation}
 \sqrt{1-c},\quad\sqrt{1-c},\quad1-c,
 \label{eq:Milne-max-axes-main}
\end{equation}
and volume
\begin{equation}
 V_{\max}(c)=\frac{4\pi}{3}(1-c)^2.
 \label{eq:Milne-max-volume-main}
\end{equation}
The following corollary relates this geometric boundary directly to the inverse-eigenvalue inequalities.

\begin{corollary}[Geometric form of the inverse theorem]\label{cor:qse-main}
For the $X$-state representative in Eq.~\eqref{eq:X-state-main}, define
\begin{equation}
 \nu=\frac{qy}{xz}.
 \label{eq:nu-main}
\end{equation}
Then
\begin{equation}
 q\le y\quad\Longleftrightarrow\quad 0\le c\le1,
 \label{eq:first-geometric-main}
\end{equation}
while
\begin{equation}
 qy\le xz
 \quad\Longleftrightarrow\quad
 0\le\nu\le1
 \quad\Longleftrightarrow\quad
 \frac{V_X}{V_{\max}(c)}\le1.
 \label{eq:second-geometric-main}
\end{equation}
Furthermore,
\begin{equation}
 \frac{V_X}{V_{\max}(c)}
 =\frac{s_\perp^2}{1-c}
 =\frac{qy}{xz}=\nu.
 \label{eq:volume-ratio-main}
\end{equation}
\end{corollary}

The first equivalence follows from $c=(y-q)/(1-z)$ and $x=(1-z)(1-c)$, while Eqs.~\eqref{eq:transverse-axis-main} and \eqref{eq:longitudinal-axis-main} give $s_\perp^2/(1-c)=qy/(xz)$. Geometrically, $c$ fixes the position of the steering-ellipsoid center, while $qy\le xz$ bounds its transverse size at that center. Saturation $qy=xz$ gives the maximal-volume steering ellipsoid allowed at the fixed center. If $q=y$ is also saturated, then $c=0$ and the ellipsoid becomes the full Bloch ball, corresponding to an entangled pure state.

For the entire $X$-state family, local filtering on Bob's side makes
his marginal maximally mixed. Up to local unitaries, the resulting state is the Choi state of an
amplitude-damping channel with damping parameter $c$, followed by
dephasing with coherence factor $\sqrt{\nu}$. In Bloch coordinates,
the combined channel acts as $(r_x,r_y,r_z)\mapsto
(\sqrt{\nu(1-c)}\,r_x,\sqrt{\nu(1-c)}\,r_y,(1-c)r_z-c)$.

\section{Exact negativity range from two PT moments}
\label{sec:moments}

In this section, we apply Theorem~\ref{thm:inverse-main} to derive
sharp bounds on the negativity of a two-qubit state. Consider an experiment in which only $p_2$ and $p_3$ are obtained.
\begin{equation}
 u=p_2=\Tr[(\rho^{\Gamma_B})^2],
 \qquad
 v=p_3=\Tr[(\rho^{\Gamma_B})^3].
 \label{eq:u-v-main}
\end{equation}
The second moment equals the ordinary purity $u=\Tr\rho^2$, whereas the third moment $v$ is sensitive to the signs of the eigenvalues of $\rho^{\Gamma_B}$. Define
\begin{equation}
 \cC(u,v)=\left\{\rho\in\cD(\bbC^2\otimes\bbC^2):
 p_2(\rho)=u,\ p_3(\rho)=v\right\},
 \label{eq:moment-compatible-set-main}
\end{equation}
and, whenever this set is nonempty,
\begin{align}
 \cN_{\min}(u,v)&=\min_{\rho\in\cC(u,v)}\cN(\rho),
 \label{eq:Nmin-main}\\
 \cN_{\max}(u,v)&=\max_{\rho\in\cC(u,v)}\cN(\rho).
 \label{eq:Nmax-main}
\end{align}
Compactness guarantees that both extrema are attained. These bounds treat $p_2$ and $p_3$ as given. If an experiment yields
a joint confidence region for the two moments, taking the smallest
$\cN_{\min}$ and the largest $\cN_{\max}$ over all compatible moment
pairs in that region gives a negativity interval with at least the
same coverage.

The measured pair $(u,v)=(p_2,p_3)$ is compatible with a two-qubit
state if and only if $0\le R\le3$ and $|K_3|\le(R/3)^{3/2}$, where
$R=4u-1$ and $K_3=(8v-6u+1)/3$~\cite{Zhang2022,Zhang2023}.
For each such pair, we determine the minimum and maximum negativity
over all two-qubit states with these moments.

\subsection{Eigenvalues at fixed second and third moments}

For an ordered candidate eigenvalue list $(x,y,z,-q)$ with one negative entry, the three positive eigenvalues obey
\begin{align}
 x+y+z&=1+q,\\
 xy+xz+yz&=A_q,
 \label{eq:Aq-main}\\
 xyz&=B_q,
 \label{eq:Bq-main}
\end{align}
where
\begin{equation}
 A_q=\frac{1+2q+2q^2-u}{2},
 \label{eq:Aq-explicit-main}
\end{equation}
\begin{equation}
 B_q=\frac{1-3u+2v+3q(1-u)+6q^2+6q^3}{6}.
 \label{eq:Bq-explicit-main}
\end{equation}
Thus $x,y,z$ are the roots of
\begin{equation}
 f_q(t)=t^3-(1+q)t^2+A_qt-B_q.
 \label{eq:cubic-main}
\end{equation}
For fixed $(u,v)$, the ordered candidate eigenvalue lists can be parametrized by $q$: the remaining eigenvalues are the roots of Eq.~\eqref{eq:cubic-main}. Realizability requires $x\ge y\ge z\ge0$ together with $q\le y$ and $qy\le xz$.

Define
\begin{equation}
 \alpha=y-q,
 \qquad
 \delta=xz-qy.
 \label{eq:slacks-main}
\end{equation}
The second and third moments become
\begin{align}
 u&=1-2\alpha+2\alpha^2-2\delta,
 \label{eq:p2-slacks-main}\\
 v&=1-3\alpha+3\alpha^2-3(1-\alpha)\delta
 -3(1-2\alpha)q(q+\alpha),
 \label{eq:p3-slacks-main}
\end{align}
and the inverse-eigenvalue inequalities reduce to
\begin{equation}
 \alpha\ge0,\qquad \delta\ge0.
 \label{eq:slack-physicality-main}
\end{equation}

At fixed $(u,v)$, when $x>y>z>0$,
\begin{equation}
 \frac{dq}{d\alpha}
 =-\frac{(x-y)(y-z)}
 {(1-2\alpha)(2q+\alpha)}<0,
 \label{eq:q-alpha-monotonic-main}
\end{equation}
while $d\delta/d\alpha=2\alpha-1<0$. Hence increasing $q$ decreases $\alpha$ and increases $\delta$. Appendix~\ref{app:fiber} proves connectedness and treats repeated roots. The following theorem gives the endpoints.

\begin{theorem}[Negativity extrema from $p_2$ and $p_3$]
\label{thm:moment-main}
Assume $\cC(u,v)\neq\varnothing$. If $\cC(u,v)$ contains NPT states, the set of positive negativity values
\[
\bigl\{\cN(\rho):\rho\in\cC(u,v),\ \cN(\rho)>0\bigr\}
\]
forms an interval. Its upper endpoint is attained at either
\begin{equation}
 y=z,\qquad q=y,
 \label{eq:upper-endpoints-main}
\end{equation}
while, if $\cC(u,v)$ contains no PPT state, its lower endpoint is positive and is attained at either
\begin{equation}
 x=y,\qquad qy=xz.
 \label{eq:lower-endpoints-main}
\end{equation}
Hence $\cN_{\min}(u,v)=0$ whenever $\cC(u,v)$ contains a PPT state, and
$\cN_{\max}(u,v)=0$ when it contains no NPT state.
\end{theorem}

The four extrema can be evaluated algebraically. For a repeated positive eigenvalue $t$, the root-collision condition is
\begin{equation}
 4t^3-3t^2+(1-u)t-\frac{1-3u+2v}{6}=0.
 \label{eq:collision-cubic-main}
\end{equation}
Defining
\begin{equation}
 D_t:=2u-1+4t-8t^2,
 \label{eq:Dt-main}
\end{equation}
the remaining eigenvalues are
\begin{equation}
 q(t)=t-\frac12+\frac12\sqrt{D_t},
 \qquad
 s(t)=\frac12-t+\frac12\sqrt{D_t}.
 \label{eq:collision-q-s-main}
\end{equation}
The ordering $0<q(t)\le s(t)\le t$ selects the $x=y$ branch, whereas $0<q(t)\le t\le s(t)$ selects the $y=z$ branch.

On the boundary $qy=xz$ with $q>0$, one has $u>1/2$. Define
\begin{equation}
 d=\sqrt{2u-1},\qquad
 \alpha_\delta=\frac{1-d}{2},\qquad
 W_\delta=\frac{3u-1-2v}{6d},
 \label{eq:delta-wall-alpha-main}
\end{equation}
which gives
\begin{equation}
 q_\delta=
 \frac{-\alpha_\delta+\sqrt{\alpha_\delta^2+4W_\delta}}{2}.
 \label{eq:qdelta-main}
\end{equation}
On $q=y$,
\begin{equation}
 q_\alpha=\sqrt{\frac{3u-1-2v}{6}},
 \label{eq:qalpha-main}
\end{equation}
with $q_\alpha>0$ and $q_\alpha(1-2q_\alpha)\ge(1-u)/2$. Appendix~\ref{app:fiber} gives the remaining eigenvalues and the
conditions for retaining each candidate.

Since every nonnegative trace-one four-point spectrum is realized as $\operatorname{spec}(\rho^{\Gamma_B})$ by a PPT state, the existence of a PPT state consistent with $(u,v)$ is equivalent to
\begin{equation}
 \frac14\le u\le1,
 \qquad
 L_4(u)\le v\le U_4(u),
 \label{eq:PPT-moment-region-main}
\end{equation}
where
\begin{equation}
 U_4(u)=
 \frac{3(6u-1)+\sqrt3(4u-1)^{3/2}}{24},
 \label{eq:U4-main}
\end{equation}
and
\begin{equation}
 L_4(u)=
 \begin{cases}
 \dfrac{3(6u-1)-\sqrt3(4u-1)^{3/2}}{24},
 &\frac14\le u\le\frac13,\\[2mm]
 \dfrac{2(9u-2)-\sqrt2(3u-1)^{3/2}}{18},
 &\frac13\le u\le\frac12,\\[2mm]
 \dfrac{3u-1}{2},
 &\frac12\le u\le1.
 \end{cases}
 \label{eq:L4-main}
\end{equation}
Equation~\eqref{eq:PPT-moment-region-main} is therefore equivalent to $\cN_{\min}(u,v)=0$.

\subsection{Geometric meaning of the moment extrema}

In terms of the steering-ellipsoid center coordinate $c$ and relative
volume $\nu$, the quantities $\alpha$ and $\delta$ are
\begin{equation}
 \alpha=y-q=(1-z)c,\qquad
 \delta=xz-qy=xz(1-\nu).
 \label{eq:slacks-geometric-main}
\end{equation}
so
\begin{equation}
 q=y\Longleftrightarrow c=0,\qquad
 qy=xz\Longleftrightarrow\nu=1.
 \label{eq:QSE-walls-main}
\end{equation}
If the minimum occurs at $\nu=1$, the steering ellipsoid reaches the
maximal volume allowed at fixed center; decreasing $q$ further would
require $\nu>1$ and violate complete positivity of the associated
channel. If the maximum occurs at $c=0$, then $y-q=0$, and the transverse
coherence supports the largest negative eigenvalue allowed by the fixed
moments. The alternatives $x=y$ and $y=z$ instead arise from degeneracies
among the positive eigenvalues of $\rho^{\Gamma_B}$.

\subsection{Comparison with other methods}
\label{subsec:comparison}

At fixed $(u,v)=(p_2,p_3)$, we compare five methods for bounding the
negativity of an unknown two-qubit state. For each method
$X\in\{\mathrm{DFN},\mathrm{YIG},\mathrm{V},\mathrm{BP},\mathrm{PT}\}$,
let $L_X$ and $U_X$ denote the resulting lower and upper bounds.
For DFN, we apply the degree-three polynomial bounds of De las Cuevas,
Fritz, and Netzer~\cite{DFN2020} on the spectral interval determined
by the purity $u$. For YIG, we apply the nonnegative moment criterion of Yu, Imai, and
G\"uhne~\cite{YuImaiGuhne2021} to the probability vector with entries
$|\lambda_i|/\sum_j|\lambda_j|$, where $\lambda_i$ are the eigenvalues
of $\rho^{\Gamma_B}$. For $\mathrm{V}$, we minimize and maximize $q$ over the ordered
eigenvalue lists considered by Zhang et al.~\cite{Zhang2023},
subject to the specified values of $p_2$ and $p_3$.
For $\mathrm{BP}$, we additionally impose $q\le y$ and
$q^2\le xz$~\cite[Theorem~3]{JohnstonPatterson2018}.
For $\mathrm{PT}$, we apply Theorem~\ref{thm:inverse-main} and impose
$q\le y$ and $qy\le xz$, so that each admissible list occurs as
$\operatorname{spec}(\rho^{\Gamma_B})$ for some two-qubit state $\rho$.
All three methods include zero as a possible negativity whenever a
nonnegative trace-one spectrum is compatible with $(u,v)$.
The following corollary identifies exactly when strengthening
$q^2\le xz$ to $qy\le xz$ raises the lower bound on negativity and
shows that the upper bounds coincide.

\begin{corollary}[When $qy\le xz$ raises the lower bound on negativity]
\label{cor:gain}
Let $u=p_2$ and $v=p_3$ be the second and third moments of
$\rho^{\Gamma_B}$ for a two-qubit state $\rho$. Then
$U_{\rm BP}=U_{\rm PT}$. If $1/2<u<1$, then
\begin{equation}
 L_{\rm PT}>L_{\rm BP}
 \quad\Longleftrightarrow\quad
 \frac14<v<\frac{3u-1}{2}.
 \label{eq:gain-strip}
\end{equation}
At $u\le1/2$ and at $u=1$ the two lower endpoints also agree.
\end{corollary}
For an NPT state with eigenvalues $(x,y,z,-q)$ of $\rho^{\Gamma_B}$,
the negativity is $q$. At fixed $(u,v)$, $xz-qy$ increases with $q$.
Starting from a list realized by a two-qubit state, every BP candidate
with a larger $q$ therefore satisfies $qy\le xz$ as well.
In the region specified by
Eq.~\eqref{eq:gain-strip}, decreasing $q$ reaches $qy=xz$ while
$q<y$ and $q^2<xz$. A further small decrease in $q$ still satisfies
both block-positivity inequalities but violates $qy\le xz$. This gives
$L_{\rm BP}<L_{\rm PT}$. Appendix~\ref{app:fiber} derives this region
and treats the case in which the same moments admit a PPT state.

\begin{figure*}[t]
 \centering
 \includegraphics[width=.96\textwidth]{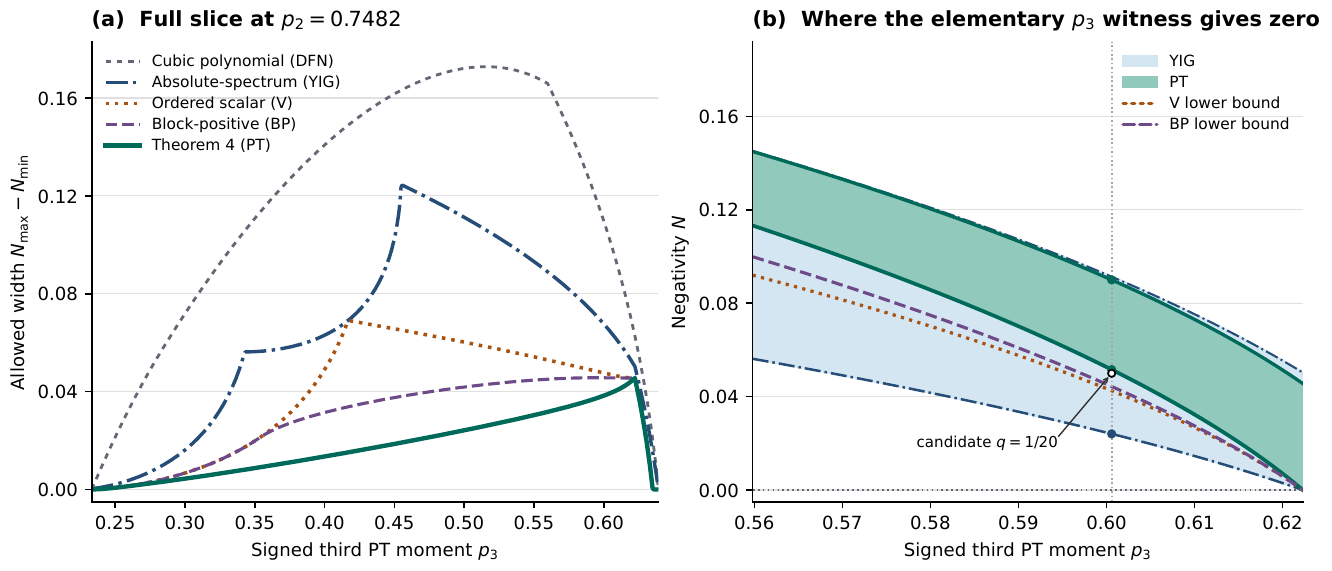}
\caption{Negativity bounds at fixed $p_2=0.7482$ as $p_3$ varies.
(a) Interval widths $U_X-L_X$ for DFN~\cite{DFN2020},
YIG~\cite{YuImaiGuhne2021}, V~\cite{Zhang2023},
BP~\cite{JohnstonPatterson2018}, and Theorem~\ref{thm:inverse-main} (PT).
(b) For $p_3\ge p_2^2$, the green band shows the exact negativity
range determined by Theorem~\ref{thm:inverse-main}. Its upper edge coincides with the BP upper bound.
At $p_3=0.60058$, block positivity allows $q=1/20$ (open circle),
below the minimum negativity attainable with these moments.}
 \label{fig:moment-range-main}
\end{figure*}

Figure~\ref{fig:moment-range-main}(a) compares the negativity interval
widths $U_X-L_X$ obtained from all five methods as $v=p_3$ varies at
fixed $u=p_2$. We choose $u=0.7482$, the second moment of the list
$(21/25,1/5,1/100,-1/20)$ considered in Sec.~\ref{sec:two-qubit-qse}.
Appendix~\ref{subsec:comparison-formulas} gives the formulas used for
DFN~\cite{DFN2020} and YIG~\cite{BerrySanders2003,YuImaiGuhne2021}.
At $v=0.45$, block positivity gives the interval $[0.202077,0.239654]$,
while the exact negativity range for two-qubit states with these moments
is $[0.221050,0.239654]$ (endpoints rounded). The upper bounds coincide,
and the increase in the lower bound reduces the interval width by $50.5\%$.

Figure~\ref{fig:moment-range-main}(b) shows how much negativity the
two moments guarantee even when the criterion $p_3<p_2^2$ does not
detect entanglement\cite{Elben2020}. At each $p_3$, the green band gives the exact
range of negativity compatible with the two moments. Its upper edge coincides with the BP upper bound. The gap between
its lower edge and the dashed BP curve shows how much the condition
$qy\le xz$ raises the lower bound. At the marked value $p_3=0.60058$,
the list used to choose $p_2$ gives $q=1/20$, which lies above
$L_{\rm BP}$ but below the minimum negativity attainable by any
two-qubit state with these moments.

\subsection{A practical example for the negativity bounds}
\label{subsec:loss}

Single-photon path entanglement has been distributed between two remote
nodes connected by a 2-km optical fiber link~\cite{Caspar2020}.
Consider a single excitation shared coherently between two paths with
a $3:1$ population ratio. It survives with probability $\eta$, giving
\begin{equation}
 \begin{aligned}
 \ket{\psi}&=\big(\sqrt{3}\ket{01}+\ket{10}\big)/2,\\
 \rho_\eta&=\eta\ket{\psi}\!\bra{\psi}+(1-\eta)\ket{00}\!\bra{00}.
 \end{aligned}
 \label{eq:lossy-source}
\end{equation}
This generalized Horodecki state can be obtained by local amplitude
damping~\cite{HorstBartkiewiczMiranowicz2013}. We use $\rho_\eta$ to
calculate $p_2$ and $p_3$, then bound the negativity over all two-qubit
states with those moments.

At $\eta=8/9$, the eigenvalues of $\rho_\eta^{\Gamma_B}$ are
$(6,4,2,-3)/9$. The two moments and the exact negativity range are
\begin{equation}
 (p_2,p_3)=\left(\frac{65}{81},\frac{29}{81}\right),\qquad
 \frac{1}{3}\le\cN\le\frac{2\sqrt{21}}{27}.
 \label{eq:loss-endpoints}
\end{equation}
The state $\rho_{8/9}$ has negativity $1/3$, the minimum compatible
with these two moments. The eigenvalues of $\rho_{8/9}^{\Gamma_B}$
satisfy $qy=xz=12/81$.
For the $X$ state constructed in Eq.~\eqref{eq:X-state-main}, this
equality gives the maximum steering-ellipsoid volume at fixed center,
as shown in Fig.~\ref{fig:spectral-QSE-main}. Since $q^2=9/81<xz$
and $q<y$, block positivity still allows smaller $q$ at these moments.
At the upper endpoint, $q=y$. Appendix~\ref{app:fiber} gives a state
supported on $\ket{01}$ and $\ket{10}$ with the same two moments and
coherence $\lvert\bra{01}\rho\ket{10}\rvert=2\sqrt{21}/27$.
Thus, even without assuming the source model, the two moments confine
the negativity to an interval of width approximately $0.00612$.

\begin{table}[tb]
 \caption{Negativity bounds from the two moments in Eq.~\eqref{eq:loss-endpoints}
at $\eta=8/9$. The BP and PT rows show how replacing $q^2\le xz$ with
$qy\le xz$ raises the lower bound. Interval widths are calculated
before rounding.}
 \label{tab:loss}
 \begin{ruledtabular}
 \begin{tabular}{lrrr}
 Constraints & Lower & Upper & Width\\
 \hline
 Cubic polynomials (DFN) & 0.24975 & 0.35722 & 0.10747\\
 Absolute spectrum (YIG) & 0.30160 & 0.33960 & 0.03800\\
 Ordered spectrum (V) & 0.32410 & 0.33956 & 0.01546\\
 Block-positive (BP) & 0.32530 & 0.33945 & 0.01415\\
 Theorem~\ref{thm:inverse-main} (PT) & $1/3$ & 0.33945 & 0.00612\\
 \end{tabular}
 \end{ruledtabular}
\end{table}

At $\eta=8/9$, the exact negativity interval is $56.8\%$ narrower than
the BP interval (Table~\ref{tab:loss}). For $0.8\le\eta\le1$,
integrating the interval width uniformly over $\eta$ gives reductions
of $49.7\%$, $66.9\%$, and $86.6\%$ relative to BP, V, and YIG,
respectively.

\section{Bounds on other resources from the same two moments}
\label{sec:resources}

We now extend the approach  used to obtain sharp negativity bounds to derive attainable bounds
for the maximal CHSH value
and directional coherent information at fixed $p_2$ and $p_3$.  In each case, we
seek the minimum and maximum over the same set $\cC(u,v)$ of
two-qubit states.

\subsection{Maximal CHSH value}

Let $s_1\ge s_2\ge s_3$ be the singular values of the correlation
matrix $M$ in Eq.~\eqref{eq:Fano-main}. The maximal CHSH value of
$\rho$ is $\beta_{\rm CHSH}=2\sqrt{s_1^2+s_2^2}$~\cite{HorodeckiCHSH1995}.
The lower and upper bounds on the maximal CHSH value therefore follow
by minimizing and maximizing $s_1^2+s_2^2$, respectively, over
correlation matrices compatible with the measured moments $p_2$ and $p_3$. Appendix~\ref{subapp:correlations} gives explicit conditions
for this optimization.

For $1/2\le u\le1$ and $1/4\le v\le(3u-1)/2$, the upper bound has
a simple form. With $\ell=(3u-1-2v)/6$,
\begin{equation}
 (\beta_{\rm CHSH})_+=2\sqrt{1+4\ell}.
 \label{eq:beta-strip}
\end{equation}
In the single-excitation example with survival probability $\eta=8/9$ in Eq.~\eqref{eq:loss-endpoints},
the state attaining the upper negativity bound also attains this upper bound on $\beta_{\rm CHSH}$. The lower bound of CHSH value can be zero. For example, the product state
$(I_2/2)\otimes(I_2+bZ)/2$ has $M=0$ and hence attains
$(\beta_{\rm CHSH})_-=0$ at $u=(1+b^2)/4$ and $v=(1+3b^2)/16$
for $0\le b\le1$.

\subsection{Directional coherent information}

The directional coherent information is defined by
$I(A\rangle B)=S(\rho_B)-S(\rho)$.
A positive value gives a lower bound on the one-way distillable
entanglement from $A$ to $B$~\cite{DevetakWinter2005}.
If the measured $p_2$ and $p_3$ imply $I(A\rangle B)\ge r>0$,
then every two-qubit state compatible with these moments has a
one-way distillable entanglement of at least $r$ entanglement bits per copy,
with classical communication from $A$ to $B$~\cite{DevetakWinter2005}. For each compatible choice of the eigenvalues of $\rho$, explicit
$X$ states are constructed to attain the minimum and maximum of $S(\rho_B)$
(See more details in Appendix~\ref{subapp:information}).

From $p_2$ and $p_3$, we can determine the sharp lower and upper bounds
on $I(A\rangle B)$, although its value is generally not uniquely
determined. For example, at $(p_2,p_3)=(1/2,1/4)$, the sharp bounds
are $-1\le I(A\rangle B)\le 0$. Subadditivity gives
$I(A\rangle B)\ge -S(\rho_A)\ge -1$, while
$S(\rho)\ge -\log_2 p_2=1$ and $S(\rho_B)\le 1$ give the upper bound. The state $\rho=(I_2/2)\otimes\ket{0}\!\bra{0}$ and the state
$\rho'=\ket{0}\!\bra{0}\otimes(I_2/2)$ obtained by interchanging
its two reduced states both have joint entropy $1$. Their partial
transposes have the same eigenvalues $(1/2,1/2,0,0)$ and hence the
stated moments. Bob's reduced state is pure for $\rho$, giving
$I(A\rangle B)=0-1=-1$, and maximally mixed for $\rho'$, giving
$I(A\rangle B)=1-1=0$.

\section{Discussion and outlook}
\label{sec:discussion}

By solving the inverse eigenvalue problem for partial
transposes of two-qubit states, we determined the sharp negativity bounds when experimental
limitations allow only $p_2$ and $p_3$ to be estimated reliably. Our result gives the substantially more narrowly negativity interval than those obtained from other proposed methods. The same approach could also be used to bound other properties like the concurrence and
the fully entangled fraction of an unknown two-qubit state at fixed
$p_2$ and $p_3$.

In higher dimensions, solving the corresponding inverse eigenvalue
problem could yield sharp negativity bounds from low-order PT moments
and provide a starting point for bounding other quantities. However,
the proof of two-qubit version relies on a special feature of the two-copy
decomposition: the antisymmetric subspace of two copies of each local
qubit is one-dimensional. For a subsystem of dimension $d>2$, the corresponding antisymmetric
subspace has dimension $d(d-1)/2$, so the similar method in our work is not available anymore. How to obtain necessary
and sufficient conditions on the inverse eigenvalues problem in
this high-dimensional setting remains an open question.

\begin{acknowledgments}
This work is supported by the National Natural Science Foundation of
China (Grants 12475020 and 92565111), Quantum Science and
Technology-National Science and Technology Major Project (2021ZD0301701),
and the National Key Research and Development Program of China
(2023YFC2205802).

The authors used ChatGPT and Codex (OpenAI) to assist with literature
searches, manuscript editing and preparation of
the accompanying code. The authors checked the references and verified
the numerical results against the analytical expressions in the paper.
All authors take responsibility for the scientific content.
\end{acknowledgments}

\section*{Data availability statement}

The supplementary material contains the data and code for reproducing
Figs.~\ref{fig:spectral-QSE-main} and~\ref{fig:moment-range-main},
Table~\ref{tab:loss}, and the comparison of negativity bounds
over $0.8\le\eta\le1$.

\appendix

\section{Negativity at fixed second and third moments}
\label{app:fiber}

\subsection{Eigenvalue lists giving the endpoints}

Fix a compatible pair $(u,v)=(p_2,p_3)$ and set $j=(1-u)/2$ and
$\ell=(3u-1-2v)/6$. By Newton's identities, the eigenvalues of
$\rho^{\Gamma_B}$ satisfy
\begin{equation}
 g(t)+h=0,\qquad g(t)=t^4-t^3+jt^2+\ell t,
 \label{eq:quartic-app}
\end{equation}
where only $h=\det(\rho^{\Gamma_B})$ can vary. The levels $-h$
giving four real roots, counted with multiplicity, form a closed
interval. At either end, two roots coincide.

Write the ordered roots as $(x,y,z,-q)$ when $q>0$ and $z\ge0$.
At a simple root $\lambda$, $d\lambda/dh=-1/g'(\lambda)$.
The signs of $g'$ at the four roots show that as $q$ increases,
$x$ and $z$ increase and $y$ decreases. These order relations
extend to repeated roots by continuity. Thus $\alpha=y-q$ decreases,
and since $\alpha\le(1-q)/2<1/2$,
\begin{equation}
 \delta=xz-qy=j-\alpha(1-\alpha)
 \label{eq:delta-app}
\end{equation}
increases. The conditions $\alpha\ge0$ and $\delta\ge0$ therefore
select an interval of $q$. If the same moments admit both a PPT
state and an NPT state, the smallest root passes through zero with
$\alpha=y\ge0$ and $\delta=xz\ge0$. The interval then extends to
zero. Equation~\eqref{eq:X-state-main} gives a two-qubit state for
each positive negativity in the interval.

At a positive endpoint, either two roots coincide or one of the
inequalities $q\le y$ and $qy\le xz$ is saturated. With distinct
roots and both inequalities strict, $q$ can still vary in either
direction. The following candidates therefore suffice.
\begin{enumerate}
\item For each real root $a$ of $g'(a)=0$, take the four entries
$a,a,r_+,r_-$, where
\[
 r_\pm=\frac{1-2a\pm\sqrt{2u-1+4a-8a^2}}{2}.
\]

\item At $q=y$, take $(x,q,z,-q)$ with $q=\sqrt{\ell}>0$ and
$x,z=(1\pm\sqrt{2u-1-4\ell})/2$.

\item At $xz=qy$ with $q>0$, put $d=\sqrt{2u-1}>0$,
$a=(1-d)/2$, and $T=\ell/d$. Then
\[
 \begin{gathered}
 q=\frac{\sqrt{a^2+4T}-a}{2},\qquad y=q+a,\\
 x,z=\frac{1-a\pm\sqrt{(1-a)^2-4T}}{2}.
 \end{gathered}
\]

\item Include zero negativity if $(u,v)$ satisfies
Eq.~\eqref{eq:PPT-moment-region-main}.
\end{enumerate}
In the first three cases, keep only real lists with the prescribed
trace and moments that satisfy Theorem~\ref{thm:inverse-main}.
Sort the entries in case 1. In cases 2 and 3, check
$x\ge y\ge z\ge0$ with the entries as assigned. The smallest and
largest retained negativities give $\cN_{\min}(u,v)$ and
$\cN_{\max}(u,v)$.

\subsection{Proof of Corollary~\ref{cor:gain}}

Suppose a two-qubit state attains $q_0>0$ at $(u,v)$. Every larger
$q$ allowed by BP satisfies $q\le y$ and, by Eq.~\eqref{eq:delta-app},
$\delta(q)\ge\delta(q_0)\ge0$. A two-qubit state therefore attains
this $q$ as well. If the moments admit a PPT state and BP allows
$q>0$, the same argument starts where the smallest root is zero.
Hence $U_{\rm BP}=U_{\rm PT}$.

Suppose now that $L_{\rm BP}<L_{\rm PT}$. Between a smaller $q$
allowed by BP and the positive minimum attained by a two-qubit
state, the ordered roots satisfy $z\ge0$ and $y\ge q$. If
$\delta>0$ at the minimum, a slightly smaller $q$ would still be
attainable. Thus the minimum occurs at $xz=qy$.

For $1/2<u<1$, set $d=\sqrt{2u-1}$ and $a=(1-d)/2$. At $xz=qy$,
we have $y=q+a$, and $x,z$ are the roots of
$t^2-(1-a)t+q(q+a)=0$. At $t=y$ this polynomial equals
$(q+a)(2q-d)$, so $z\le y\le x$ precisely when $0\le q\le d/2$.
On this boundary,
\begin{equation}
 v=\frac{3u-1}{2}-3d\,q(q+a).
 \label{eq:gain-wall-param}
\end{equation}
For $0<q<d/2$, the roots are distinct, $q<y$, and $q^2<xz=qy$.
At fixed $u,v$, decreasing $q$ slightly preserves both BP inequalities
but makes $\delta<0$. A nonnegative list is also excluded, since its
third elementary symmetric function would be $-\ell=-dq(q+a)<0$.

As $q$ runs from $0$ to $d/2$ along Eq.~\eqref{eq:gain-wall-param},
$v$ decreases from $(3u-1)/2$ to $1/4$. At $q=0$ the minimum is
already zero. At $q=d/2$ the two largest roots coincide, so the
ordering precludes a smaller $q$. Thus $L_{\rm PT}>L_{\rm BP}$
precisely on the open interval in Eq.~\eqref{eq:gain-strip}.

For $u\le1/2$, $j\ge1/4$ gives $\delta\ge0$ directly. For $u=1$,
every BP candidate with $q>0$ satisfies
$0=j=(xz-q^2)+(y-q)(1-y)$. Both terms are nonnegative and $y<1$,
so $q=y$ and $xz=q^2=qy$.

In the example of Eq.~\eqref{eq:lossy-source}, $\rho_{8/9}$ attains
the lower endpoint $q=1/3$ at $xz=qy$. At the upper endpoint,
$q=y=\sqrt{\ell}=2\sqrt{21}/27$. Take $x,z$ from the $q=y$ case
above. The state supported on $\ket{01},\ket{10}$ with block
\[
 \begin{pmatrix}x&\sqrt{\ell}\\ \sqrt{\ell}&z\end{pmatrix}
\]
has the same moments. The block has trace one and determinant
$j=8/81$. The partial transpose of the state has eigenvalues
$x,z,\pm\sqrt{\ell}$, so its negativity attains the upper bound.

\subsection{Bounds used in the comparison}
\label{subsec:comparison-formulas}

For DFN~\cite{DFN2020}, let $f(t)=\max(0,-t)$. Summing $f$ over
the eigenvalues of $\rho^{\Gamma_B}$ gives the negativity. If one
eigenvalue is $t$, the other three sum to $1-t$, so
$u\ge t^2+(1-t)^2/3$. Thus $t\in[a_u,b_u]$, where
$a_u,b_u=(1\mp\sqrt{3(4u-1)})/4$. For a polynomial
$c(t)=\sum_{i=0}^3c_it^i$ of degree at most three, summing over
the same eigenvalues gives $T(c)=4c_0+c_1+uc_2+vc_3$. The lower
and upper bounds are $\sup_{c\le f}T(c)$ and $\inf_{c\ge f}T(c)$,
respectively, with both inequalities imposed throughout $[a_u,b_u]$.

For YIG~\cite{BerrySanders2003,YuImaiGuhne2021}, take a candidate
negativity $n\in[0,\max(0,-a_u)]$. If $\rho$ has negativity $n$,
the absolute values of the eigenvalues of $\rho^{\Gamma_B}$ sum
to $1+2n$. Dividing by this sum gives a probability vector with
second and third moments
\[
 r=\frac{u}{(1+2n)^2},\qquad t=\frac{v+2n^3}{(1+2n)^3}.
\]
The YIG bounds are the smallest and largest $n$ satisfying
$1/4\le r\le1$ and $L_4(r)\le t\le U_4(r)$. These conditions test
whether a probability vector has moments $r,t$. For $n>0$, a vector
obtained from $\rho^{\Gamma_B}$ must also contain the entry $n/(1+2n)$.

For V, BP, and PT, vary $h$ in Eq.~\eqref{eq:quartic-app} subject
to the respective conditions in Sec.~\ref{subsec:comparison}. For
BP, the boundary $xz=q^2$ gives
\begin{equation}
 (q^2+jq-\ell)(q^3+jq-\ell)+jq^4=0.
 \label{eq:bp-poly}
\end{equation}
For each positive root, set $\alpha=1+j/q-\ell/q^2$ and $y=q+\alpha$.
Recover $x,z$ from $x+z=1-\alpha$ and $xz=q^2$. Keep the list if
it satisfies the ordering, the moment constraints, and both BP
inequalities.

\section{Attaining states for the CHSH value and coherent information}
\label{app:resources}

\subsection{Maximal CHSH value}
\label{subapp:correlations}

In the magic basis $(\Phi^+,i\Phi^-,i\Psi^+,\Psi^-)$, where
$\Phi^\pm=(\ket{00}\pm\ket{11})/\sqrt2$ and
$\Psi^\pm=(\ket{01}\pm\ket{10})/\sqrt2$, local unitaries act as
real rotations up to an overall phase. They can therefore diagonalize
the real part of $\rho$. Write $\rho=D+iB$, where
$D=\diag(d_1,d_2,d_3,d_4)$ with $d_1\ge\cdots\ge d_4$ and
$B=(b_{ij})$ is real and skew-symmetric. Set
\[
 t_i=2(d_1+d_{i+1})-1,\qquad i=1,2,3.
\]
The ordered singular values of $M$ are $t_1,t_2,|t_3|$. To find the
two extrema of $\beta_{\rm CHSH}$, minimize and maximize
$2\sqrt{t_1^2+t_2^2}$ subject to $D+iB\succeq0$, $\sum_i d_i=1$,
and
\begin{align}
 u&=\sum_i d_i^2+2\sum_{i<j}b_{ij}^2,\notag\\
 v&=\sum_i(1/2-d_i)^3+3\sum_{i<j}(d_i+d_j)b_{ij}^2.
 \label{eq:magic-moments}
\end{align}
The feasible set is compact, so both extrema are attained.

For each feasible $D$, the same moments can be attained with both
entries of one of the pairs $(b_{12},b_{34})$, $(b_{13},b_{24})$,
$(b_{14},b_{23})$ equal to zero. Write the pairs as columns $z_i$
and replace each product $z_i z_i^{\mathsf T}$ by a real symmetric
$2\times2$ matrix $X_i\succeq0$. Equation~\eqref{eq:magic-moments}
fixes two linear combinations of their diagonal entries, including
the total trace. These constraints also fix the third elementary
symmetric function $e_3(D)-\sum_{i<j}(1-d_i-d_j)b_{ij}^2$.
The determinant is
\[
 \det(D+iB)=\prod_i d_i-\sum_{i<j}d_kd_l b_{ij}^2
             +(\operatorname{Pf}B)^2,
\]
where $k,l$ are complementary to $i,j$ and
$\operatorname{Pf}B=b_{12}b_{34}-b_{13}b_{24}+b_{14}b_{23}$.
The Pfaffian is linear in the off-diagonal entries of the $X_i$.
Thus the determinant expression is convex on this compact set and
attains its maximum at an extreme point.

At an extreme point, the ranks $r_i=\rank X_i$ satisfy
$\sum_i r_i(r_i+1)/2\le2$. Otherwise a nonzero symmetric
perturbation on their ranges would preserve both constraints in
either direction. At most two $X_i$ are therefore nonzero, each of
rank one. Factoring these matrices gives a new $B$. This preserves
the first three elementary symmetric functions of the eigenvalues
of $D+iB$ and can only increase its determinant. All coefficients
of $\det(tI+D+iB)$ are consequently nonnegative, so this polynomial
is positive for $t>0$ and $D+iB\succeq0$. The CHSH value depends
only on $D$, so it suffices to test the three choices of a vanishing
pair.

To prove Eq.~\eqref{eq:beta-strip}, choose local axes with
$M=\diag(a,b,t)$ and $1\ge a\ge b\ge|t|$. The singlet expectation
of $\rho$ gives $a+b+t\le1$, from which
\[
 3(a^2+b^2)-t^2+2abt\le3.
\]
Indeed, for $a+b\le1$, use $-t^2+2abt\le a^2b^2$ and
$3(a^2+b^2)+a^2b^2=3(a+b)^2-ab(6-ab)\le3$.
For $a+b>1$, the left-hand side increases with $t$ on
$t\le1-a-b$. Its value at $1-a-b$ is $3-2(a+b+2)(1-a)(1-b)\le3$.

In Bloch coordinates, $R=4u-1=|\bm r|^2+|\bm s|^2+\|M\|_F^2$ and
$K_3=(8v-6u+1)/3=\bm r^{\mathsf T}M\bm s+\det M$. Since $a\le1$,
$2\bm r^{\mathsf T}M\bm s\le|\bm r|^2+|\bm s|^2$. Combining the
inequalities gives $\beta_{\rm CHSH}^2\le3+R-2K_3=4(1+4\ell)$.
For $1/2\le u\le1$ and $1/4\le v\le(3u-1)/2$, put
$x,z=(1\pm\sqrt{2u-1-4\ell})/2$. The state with block
$\bigl(\begin{smallmatrix}x&\sqrt\ell\\ \sqrt\ell&z\end{smallmatrix}\bigr)$
on $\ket{01},\ket{10}$ and zero elsewhere has the prescribed moments.
The block has determinant $xz-\ell=j\ge0$, and the correlation matrix
has singular values $1,2\sqrt\ell,2\sqrt\ell$. This state attains
the upper bound.

\subsection{Directional coherent information}
\label{subapp:information}

Fix the eigenvalues $\mu_1\ge\cdots\ge\mu_4\ge0$ of $\rho$, with
$\sum_i\mu_i=1$ and $\sum_i\mu_i^2=u$. This fixes $S(\rho)$, so
the extrema of $I(A\rangle B)$ at these eigenvalues depend only on
the range of $S(\rho_B)$. Write
\[
 \begin{gathered}
 M_3=\sum_i\mu_i^3,\qquad k=\mu_1-\mu_2+\mu_3-\mu_4,\\
 A=\mu_1-\mu_3,\qquad G=\mu_2-\mu_4.
 \end{gathered}
\]
These eigenvalues are compatible with $v$ precisely when
$Q_-\le v\le Q_+$, where
\begin{equation}
 Q_-=M_3-3kA^2/4,\qquad Q_+=M_3+3kG^2/4.
 \label{eq:ordinary-feasible}
\end{equation}
For $k>0$, set
\[
 (d,t)=
 \begin{cases}
 (G,\sqrt{A^2+4(v-M_3)/(3k)}), & v\le M_3,\\
 (A,\sqrt{G^2-4(v-M_3)/(3k)}), & v>M_3.
 \end{cases}
\]
The minimum and maximum lengths of Bob's Bloch vector are
$b_-=|d-t|$ and $b_+=d+t$. For $k=0$, $\mu=(a,a,c,c)$,
$v=M_3$, and the two lengths are $0$ and $2(a-c)$.

Bob's entropy $h_2((1+b)/2)$ decreases with the length $b$. Hence
\begin{align}
 I_-&=\min_{\mu}\bigl[h_2((1+b_+)/2)-H(\mu)\bigr],\notag\\
 I_+&=\max_{\mu}\bigl[h_2((1+b_-)/2)-H(\mu)\bigr],
 \label{eq:I-endpoints}
\end{align}
where $h_2$ is the binary entropy and $H(\mu)=-\sum_i\mu_i\log_2\mu_i$.
Here $\mu$ ranges over the ordered eigenvalues with trace one, purity
$u$, and $Q_-\le v\le Q_+$. Given $\mu_1,\mu_2$, the trace and
purity fix the ordered pair $\mu_3,\mu_4$.

To bound $b$, choose Bob's basis so that $\rho_B=(I+bZ)/2$ with
$b\ge0$. Let $w_i$ be the diagonal entries of $I_A\otimes\ket0\!\bra0$
in an eigenbasis of $\rho$. Then $0\le w_i\le1$, $\sum_iw_i=2$, and
$\sum_i\mu_iw_i=(1+b)/2$. The qubit identity
$YX^{\mathsf T}Y=(\Tr X)I-X$ makes $\rho^{\Gamma_A}$ unitarily
equivalent to $I_A\otimes\rho_B-\rho$. Taking the trace of the cube
gives
\[
 v=-\frac{1+3b^2}{4}+\frac{3(1-b)u}{2}
    +3b\sum_i\mu_i^2w_i-M_3.
\]
The vertices of $\{w:0\le w_i\le1,\ \sum_iw_i=2\}$ have two entries
equal to one. Thus $(\sum_i\mu_iw_i,\sum_i\mu_i^2w_i)$ lies in the
convex hull of $(\mu_i+\mu_j,\mu_i^2+\mu_j^2)$ for $i<j$. Its
upper edges join the points indexed by $12,14,34$, and its lower
edges join $12,23,34$. Substitution gives
\begin{align*}
 Q_-+\tfrac{3k}{4}(G-b)^2&\le v\le Q_-+\tfrac{3k}{4}(G+b)^2,\\
 Q_+-\tfrac{3k}{4}(A+b)^2&\le v\le Q_+-\tfrac{3k}{4}(A-b)^2.
\end{align*}
For $k>0$, these inequalities give $Q_-\le v\le Q_+$ and
$b_-\le b\le b_+$. For $k=0$, they give $v=M_3$, while
$(1+b)/2=\sum_i\mu_iw_i\le\mu_1+\mu_2$ gives $b\le2(a-c)$.

To attain the bounds for $k>0$, choose $(\mu_1,\mu_3)$ if $v\le M_3$
and $(\mu_2,\mu_4)$ if $v>M_3$. Let $s,\Delta$ be the sum and
difference of the chosen pair, and $p\ge p'$ the remaining eigenvalues.
Set $g=\sqrt{\Delta^2-t^2}/2$. The $X$ state
\begin{equation}
 \rho=\begin{pmatrix}
 (s+t)/2&0&0&g\\
 0&p&0&0\\
 0&0&p'&0\\
 g&0&0&(s-t)/2
 \end{pmatrix}
 \label{eq:entropy-attainer}
\end{equation}
gives $b_-=|t-(p-p')|$. Exchanging $p$ and $p'$ gives
$b_+=t+(p-p')$. Both states have eigenvalues $\mu_i$ and satisfy
$\Tr[(\rho^{\Gamma_B})^3]=M_3+3(1-2s)g^2=v$. For $k=0$, the
same construction with the pair $(a,c)$ and $t=a-c$ gives the two
lengths $0$ and $2(a-c)$. Thus both bounds in
Eq.~\eqref{eq:I-endpoints} are attained.

\bibliography{reference}

\end{document}